\documentclass[11pt,fleqn]{article}
\usepackage{times}
\usepackage{verbatim,amsmath,amssymb,amsthm,enumerate}
\usepackage[dvips]{graphicx}

\newcommand{\remove}[1]{}
\newcommand{\set}[1]{{\left\{ #1\right\}}}
\newcommand{\B}[0]{{\left\{0,1\right\}}}
\newcommand{\ol}{\overline}

\newcommand{\Tr}{\mathrm{Tr}}
\newcommand{\eps}{{\epsilon}}

\newcommand{\logeps}{\log \epsilon^{-1}}
\newcommand{\eqdef}{\stackrel{\rm def}{=}}

\newcommand{\poly}{{\rm poly}}
\newcommand{\Hom}{{\rm Hom}}
\newcommand{\opt}{{\sf opt}}

\newcommand{\ra} {\right \rangle}
\newcommand{\la} {\left  \langle}
\newcommand{\ket}[1]{\left|#1\right\rangle}
\newcommand{\bra}[1]{\left\langle #1\right|}

\newcommand{\ketbra}[2]{\ket{#1}\!\bra{#2}}
\newcommand{\tensor}{{\otimes}}

\newcommand{\vj}{{x_{j}}}
\newcommand{\wj}{{y_{j}}}

\newcommand{\norm}[1]{\left\| #1 \right\|}
\newcommand{\trnorm}[1]{\norm{#1}_{\rm tr}}
\newcommand{\dnorm}[1]{\norm{#1}_{\diamond}}

\newcommand{\cA}{\mathcal{A}}

\newcommand{\cV}{\mathcal{V}}
\newcommand{\cW}{\mathcal{W}}

\newcommand{\R}{\mathbb{R}}

\newcommand{\PP}{\text{P}}
\newcommand{\NP}{\text{NP}}
\newcommand{\IP}{\text{IP}}
\newcommand{\QIP}{\text{QIP}}
\newcommand{\PSPACE}{\text{PSPACE}}
\newcommand{\EXP}{\text{EXP}}

\newcommand{\KO}{K^{(1)}}
\newcommand{\SO}{S^{(1)}}
\newtheorem{theorem}{Theorem}[section]
\newtheorem{remark}{Remark}[section]
\newtheorem{definition}{Definition}[section]

\newtheorem{claim}{Claim}[section]

\begin{document}

\title{On the complexity of approximating the diamond norm}

\author{
Avraham Ben-Aroya\thanks{Department of Computer Science, Tel-Aviv
University, Tel-Aviv 69978, Israel. Supported by the Adams
Fellowship Program of the Israel Academy of Sciences and
Humanities, by the European Commission under the Integrated
Project QAP funded by the IST directorate as Contract Number
015848 and by USA Israel BSF grant 2004390. Email: abrhambe@post.tau.ac.il.}
\and Amnon Ta-Shma\thanks{Department of Computer Science, Tel-Aviv
University, Tel-Aviv 69978, Israel. Supported by the European Commission under the Integrated
Project QAP funded by the IST directorate as Contract Number
015848, by Israel Science Foundation grant 217/05 and by USA Israel BSF grant 2004390.
Email: amnon@tau.ac.il.
}}

\date{}
\maketitle \thispagestyle{empty}

\begin{abstract}
The \emph{diamond norm} is a norm defined over the space of
quantum transformations. This norm has a natural operational
interpretation: it measures how well one can distinguish between
two transformations by applying them to a state of arbitrarily
large dimension. This interpretation makes this norm useful in the
study of quantum interactive proof systems.

In this note we exhibit an efficient algorithm for computing this
norm using convex programming. Independently of us,
Watrous~\cite{W09} recently showed a different algorithm to
compute this norm. An immediate corollary of this algorithm is a slight
simplification of the argument of Kitaev and Watrous~\cite{KW00}
that $\QIP \subseteq \EXP$.
\end{abstract}

\section{Introduction}

How well can one distinguish two quantum transformations? Imagine
we have access to some unknown admissible super-operator $T$ and
we want to distinguish the case it is $T_1$ from the case it is
$T_2$ ($T_1$ and $T_2$ are known). Suppose that $T_1$ and $T_2$
take as input a state from a Hilbert space $\cV$. One possible
test to distinguish $T_1$ from $T_2$ is preparing an input state
$\rho \in D(\cV)$ (where $D(\cV)$ denotes the set of density
matrices over $\cV$), applying $T$ on $\rho$ and measuring the
result. This corresponds to:
$$
\sup \set{ \trnorm{T_1 \rho - T_2 \rho} ~:~ \rho \in D(\cV)}.
$$

However, somewhat surprisingly, it turns out that often one can
distinguish $T_1$ and $T_2$ better, by taking an auxiliary Hilbert
space $\cA$, preparing an \emph{entangled} input state $\rho \in
D(\cV \tensor \cA)$, applying $T$ on the $\cV$ register of $\rho$
and then measuring the global result. Therefore, we define:
\begin{eqnarray*}
{\rm dist}(\rho_1,\rho_2) & = & \sup \set{ \trnorm{(T_1 \tensor
I_{L(\cA)})\rho - (T_2 \tensor I_{L(\cA)})\rho} ~:~
\dim(\cA)<\infty,~ \rho \in D(\cV \tensor \cA)}.
\end{eqnarray*}

Kitaev~\cite{K97} proved that this phenomena is restricted by
dimension and the maximum is attained already with an auxiliary
Hilbert space $\cA$ of dimension $\dim(\cA) \le \dim(\cV)$. Define
the following functions on general (not necessarily admissible)
super-operators $T:L(\cV) \to L(\cW)$:
\begin{eqnarray*}
\trnorm{T} & = & \sup \set{ \trnorm{T(X)} : X \in L(\cV), \trnorm{X}=1} \mbox{ , and,}\\
\dnorm{T} & = & \trnorm{T \tensor I_{L(\cV)}}.
\end{eqnarray*}

\noindent Kitaev showed that both $\trnorm{\cdot}$ and
$\dnorm{\cdot}$ are norms. Furthermore, Rosgen and
Watrous~\cite[Lemma 2.4]{RW05} showed ${\rm
dist}(T_1,T_2)=\dnorm{T_1-T_2}$ for $T_1$ and $T_2$ that are
completely positive.

The diamond norm naturally appears when studying the class $\QIP$
of languages having a single-prover, multi-round interactive proof
protocol between an all-powerful prover and an efficient quantum
verifier. Kitaev and Watrous~\cite{KW00} showed that, without loss
of generality, perfect completeness can be achieved and three
rounds suffice (starting with the verifier). They also showed that
the value of a three round quantum interactive protocol can be
expressed as $\dnorm{T}$, for some super-operator $T$ that is
naturally defined given the protocol of the verifier. They used
this characterization, and the fact that $\dnorm{T_1 \tensor
T_2}=\dnorm{T_1} \cdot \dnorm{T_2}$ to show perfect parallel
amplification for $\QIP$ protocols. Finally, they showed that
$\QIP \subseteq \EXP$ by reducing the problem to an exponential
size semi-definite programming problem. Thus $\QIP$ is somewhere
between $\PSPACE$ and $\EXP$ (the containment $\PSPACE=\IP
\subseteq \QIP$ is immediate). Very recently, Jain et. al.
\cite{JJUW09} showed that $\QIP=\PSPACE$, by showing a space
efficient solution to a semi-definite program that captures the
complexity of the class $\QIP$.

Another connection between $\QIP$ and the diamond norm was given
by Rosgen and Watrous~\cite{RW05}. They defined the promise
problem ${\sf QCD}_{a,b}$ (quantum circuit distinguishability)
whose input is two admissible super-operators $T_1$ and $T_2$, the
``yes" instances are pairs $(T_1,T_2)$ for which $\dnorm{T_1-T_2}
\ge a$ and the ``no" instances are the pairs for which
$\dnorm{T_1-T_2} \le b$. Rosgen and Watrous~\cite{RW05} proved
that for every $a<b$ the problem ${\sf QCD}_{a,b}$ is
$\QIP$-complete (see also~\cite{R08}).

The work of Kitaev and Watrous, as well as the work of Rosgen and
Watrous do not imply that approximating the diamond norm itself
can be done in $\PP$. In this note we prove that the diamond norm
can be computed by solving a convex optimization problem, and
therefore it is in $\PP$. More precisely, if we are given as input
a description of $T:L(\cV) \to L(\cV)$, e.g., written as a matrix of
dimensions $N^2 \times N^2$ (where $N=\dim(\cV)$), and we are
given $\eps>0$, then we can approximate $\dnorm{T}$ to within
$\eps$ additive accuracy in time $\poly(N, \log \eps^{-1})$.
Independently of us, Watrous~\cite{W09} recently showed a similar
result using a semi-definite program.

This claim can also be used to simplify the (somewhat more
complicated) proof given in~\cite{KW00} that $\QIP \subseteq
\EXP$. To see this, notice that Kitaev and Watrous already proved
that the value of a three round quantum interactive proof system
can be captured as the
diamond norm of a natural super-operator $T$. Thus, given such a
proof system, all we need to do is to explicitly write down the
description of $T$ (which can be done in $\PSPACE$ and therefore
in time exponential in $\poly(n)$, where $n$ is the input length
of the \QIP{} protocol) and then approximate its diamond norm, in
time polynomial in $\exp(\poly(n))$.

Our proof is surprisingly simple. We use an equivalent formulation
of the diamond norm, proved by Kitaev, and we notice that it gives
a convex program using the joint concavity of the fidelity
function. We use a representation for density matrices suggested
by Liu~\cite{Liu06} in a different context for a similar purpose.

\section{Preliminaries}

Let $\cV,\cW$ be two Hilbert spaces. $\Hom(\cV,\cW)$ denotes the
set of all linear transformations from $\cV$ to $\cW$ and is a
vector space of dimension $\dim(\cV) \cdot \dim(\cW)$ equipped
with the Hilbert-Schmidt inner product $\la T_1,T_2 \ra =
\Tr(T_1^\dagger T_2)$. $L(\cV)$ denotes $\Hom(\cV,\cV)$. Let
$\set{\ket{i}}$ denote the standard basis for $\cV$. The set
$$\set{\ket{i}\bra{j} ~:~ 1 \le i,j \le \dim(\cV)}$$ is an
orthonormal basis of $L(\cV)$. When $\dim(\cV)=2^n$, tensor
products of Pauli operators form another natural basis for
$L(\cV)$. The Pauli operators are

$$\sigma_{0}=\left(
\begin{array}{cc}1 & 0 \\ 0 & 1
\end{array}\right)~~ \sigma_{1}=\left( \begin{array}{cc}0 & 1 \\ 1
& 0 \end{array}\right)~~ \sigma_{2}=\left( \begin{array}{cc}0 & -i
\\ i & 0 \end{array}\right)~~ \sigma_{3}=\left(
\begin{array}{cc}1 & 0 \\ 0 & -1 \end{array}\right).$$
The set $\set{\sigma_{i_1} \tensor \ldots \tensor \sigma_{i_n} ~:~
0 \le i_1,\ldots,i_n \le 3}$ is an orthogonal basis of $L(\cV)$,
and all basis elements have eigenvalues $\pm 1$.

For a linear operator $A \in \Hom(\cV,\cW)$, the spectral norm of $A$ is
$$\norm{A} \eqdef \sup_{x: \norm{x}=1} x^\dagger A^\dagger A x$$
and is equal to the largest singular value of $A$. For any Pauli
operator $P$, $\norm{P} =1$. The $\ell_2$ norm of $A$ is $\norm{A}_2=\Tr(A^\dagger A)$ and is equal to the $\ell_2$ norm of the singular values of $A$.

A pure state is a unit vector in some Hilbert space. A general
quantum system is in a {\em mixed state\/}--a probability
distribution over pure states. Let $\{p_i, \ket{\phi_i}\}$ denote
the mixed state in which the pure state~$\ket{\phi_i}$ occurs with
probability~$p_i$. The behavior of the mixed-state
$\set{p_i,\ket{\phi_i}}$ is completely characterized by its {\em
density matrix\/}~$\rho = \sum_i p_i \ketbra{\phi_i}{\phi_i}$, in
the sense that two mixed states with the same density matrix
behave the same under any physical operation. Notice that a
density matrix over a Hilbert space $\cV$ belongs to $L(\cV)$.
Density matrices are positive semi-definite operators and have
trace $1$. We denote the set of density matrices over $\cV$ by
$D(\cV)$.

\vspace*{0.3cm} \noindent {\bf{Trace norm and fidelity.}} The {\em
trace norm\/} of a matrix~$A$ is defined by
$$\trnorm{A} =\Tr(|A|)= \Tr\left({\sqrt{A^\dagger A}}\right),$$
which is the sum of the magnitudes of the singular values of~$A$.
One way to measure the distance between two density matrices
$\rho_1$ and $\rho_2$ is by their trace distance
$\trnorm{\rho_1-\rho_2}$. Another useful alternative to the trace
metric as a measure of closeness of density matrices is the {\em
fidelity}. For two positive semi-definite operators $\rho_1,
\rho_2$ on the same finite dimensional space~$\cV$ (not
necessarily having trace $1$) we define
$$F(\rho_1,\rho_2) \;=\; \left[\Tr\left(
                             \sqrt{{\rho_1}^{1/2}\,\rho_2\,{\rho_1}^{1/2}}
                             \right) \right]^2
                   \;\;=\;\; \trnorm{ \sqrt{\rho_1}\sqrt{\rho_2}}^2.
$$

We remark that some authors define
$\sqrt{F}=\trnorm{\sqrt{\rho_1}\sqrt{\rho_2}}$ as the fidelity.
Our definition is consistent with~\cite{KSV02}. $\sqrt{F}$ is
jointly concave, i.e., for every set
$\set{(\rho_i,\xi_i)}_{i=1}^k$ of pairs of density matrices and
every $0 \le \lambda_1,\ldots,\lambda_k \le 1$ such that
$\sum_{i=1}^k \lambda_i=1$,
$$
\sqrt{F}(\sum_{i=1}^k \lambda_i \rho_i,\sum_{i=1}^k \lambda_i
\xi_i) \;\ge\; \sum_{i=1}^k \lambda_i \sqrt{F}(\rho_i,\xi_i).
$$
A proof of this fact appears, e.g., in~\cite[Exercise
9.19]{NC00}.\footnote{Note that in~\cite{NC00} the fidelity
function is defined to be $\sqrt{F}$. In particular, the joint
concavity of the fidelity function proved in~\cite[Exercise
9.19]{NC00} proves joint concavity of $\sqrt{F}$ according to our
notation.} We remark that $F$ is not jointly concave
(see~\cite[Section 2]{MPHUZ08} for a short survey on what is known
about the fidelity function).

\vspace*{0.3cm} \noindent {\bf{The diamond norm.}} Kitaev gave a
different equivalent characterization of the diamond norm as
follows. Any $T:L(\cV) \to L(\cV)$ can be written in a
\emph{Stinespring representation}, i.e., as
$$T(X) = \Tr_\cA (B X C^\dagger),$$
where $B,C \in \Hom(\cV, \cV \tensor \cA)$ and $\dim(\cA) \le
(\dim(\cV))^2$ (see, e.g.,~\cite[page 110]{KSV02} or~\cite[Lecture
4]{W04}). Define two completely positive super-operators $T_1,
T_2: L(\cV) \to L(\cA)$:
\begin{eqnarray}
\label{eq:T1}
T_1(X) &=& \Tr_\cV (BXB^\dagger) ,\\
\label{eq:T2}
T_2(X) &=& \Tr_\cV (CXC^\dagger) .
\end{eqnarray}
Then, the diamond norm of $T$ can be written as
$$\dnorm{T} = \max \set{ \sqrt{F}(T_1(\rho), T_2(\xi)) ~:~
\rho, \xi \in D(\cV)}.$$ The proof of this characterization can be
found in~\cite[Problem 11.10]{KSV02} or in Watrous' lecture
notes~\cite[Lecture 22, Theorem 22.2]{W04} (and notice that
Watrous defines the fidelity function to be $\sqrt{F}$). Further
information on the trace norm and the diamond norm of
super-operators can be found in~\cite{KSV02}.

\vspace*{0.3cm} \noindent {\bf{Convex programming.}} Maximizing a
convex function over a convex domain is, in general, \NP{}-hard
(see~\cite{FV95} for a survey). In sharp contrast to this,
\emph{convex programming}, which is the problem of
\emph{minimizing} a convex function over a convex domain, is in
\PP{}. One of the reasons that convex programming is easier to
solve is due to the fact that in a convex program any \emph{local}
optimum equals the \emph{global} optimum. Special cases of convex
programming are semi-definite programming and linear programming.
Convex programming can be solved in polynomial time using the
ellipsoid algorithm~\cite{K79} or interior-point methods. Often,
these algorithms assume a \emph{separation oracle}, i.e., an
efficient procedure that given a point tells whether it belongs to
the convex set, and if not, gives a half-space that separates the
point from the convex set. However, the problem can also be solved
using a \emph{membership oracle}~\cite{YN76,GLS88} (a randomized
algorithm is given in~\cite{BV04}).

For $a \in \R^n$ and $R>0$ we define $B_n(a,R)=\set{x \in \R^n ~:~
\norm{x-a}_2 \le R}$. For a set $K \subseteq \R^n$ we define
\begin{eqnarray*}
K_{-\eps} & = & \set{x \in \R^n ~:~ B_n(x,\eps) \subseteq K} \\
K_{+\eps} & = & \set{x \in \R^n ~:~ \exists {y \in K} \mbox{
\;such that\; } x \in B_n(y,\eps)}
\end{eqnarray*}
That is, $K_{-\eps}$ is the set of points $\eps$-deep in $K$ and
$\R^n \setminus K_{+\eps}$ is the set of points $\eps$-deep in the
complement of $K$.

\begin{definition}
A function $O_K:\R^n \times \R^+ \to \B$ is a \emph{membership
oracle for $K \subseteq \R^n$} if for every $\eps>0$,
$O_K(x,\eps)=1$ for any $x \in K_{-\eps}$ and $O_K(x,\eps)=0$ for
any $x \not \in K_{+\eps}$. $O_K$ is \emph{efficient}, if it runs
in time polynomial in its input length.
\end{definition}
\begin{definition}
A function $O_f:K \times \R^+ \to \R$ is an \emph{evaluation
oracle computing $f$ over $K$}, if for every $x \in K$ and every
$\eps>0$, $|f(x)-O_f(x,\eps)| \le \eps$. $O_f$ is
\emph{efficient}, if it runs in time polynomial in its input
length.
\end{definition}

\begin{theorem}[\cite{YN76},{\cite[Theorem 4.3.13]{GLS88}}]
\label{thm:convex-programming} There exists an algorithm that
solves the following problem:
\begin{description}
\item [Input]: \vspace{-22pt}
\begin{enumerate}

\item
A convex body $K$ given by an efficient membership oracle.

\item An integer $n$, rational numbers $R,r>0$ and
a vector $a_0 \in \R^n$ such that
$$B_n(a_0,r) \subseteq K \subseteq B_n(\ol{0},R) \subseteq \R^n.$$

\item
A rational number $\eps>0$.

\item
\label{it:g}
A convex function $g: K_{+\eps} \to \R$ given by an efficient evaluation oracle.
\end{enumerate}

\item [Output]:
A value $x \in K_{+\eps}$ such that $|g(x)-\widetilde{\opt}| \le \eps$, where $\widetilde{\opt}=\min_{x \in
K_{-\eps}} g(x)$.
\end{description}

The algorithm runs in time $\poly(n,\logeps,\log (R/r))$.
\end{theorem}

\begin{remark}
The theorem is a slight variation of the one appearing in \cite{GLS88}. There $g$ is required to be defined and convex over the whole of $\R^n$, whereas we only require that it is defined over $K_{+\eps}$.

To see why our variation is correct, notice that the proof given in \cite{GLS88} works by a Turing reduction that queries membership in the convex set $\set{(x,b)~|~x\in K, g(x) \le b}$. If $x$ is $\eps$-far from $K$ it is also $\eps$-far from $\set{(x,b)~|~x\in K, g(x) \le b}$ and we can safely reject. Hence, we only need to query $g$ on inputs that are in $K_{+\eps}$.
\end{remark}

\section{Approximating the diamond norm in \PP}

\subsection{Representing density matrices}

We follow~\cite{Liu06} in the way we represent density matrices as
vectors. This is due to that fact that we need the set of vectors
representing the density matrices to contain and to be contained in
balls of appropriate radii around the origin.

We represent $\rho \in D(\cV)$ by its Pauli-basis coefficients,
but excluding the identity coefficient which is always $1$. Thus,
we represent $\rho \in D(\cV)$ as a vector $v(\rho) \in
\R^{N^2-1}$, where the $i$th coordinate of this vector is given by
$v_i(\rho) = \Tr(P_{i+1}\rho)$, where $P_i$ is the $i$th Pauli
operator and $P_1 = I$. (Notice that $\Tr(P\rho) \in \R$ for
Hermitian $P$ and positive semi-definite $\rho$.) We let
$$\KO=\set{ v(\rho) ~:~  \rho \in D(\cV)}.$$
The converse transformation $\Phi:\KO \to D(\cV)$ is defined by
$$\Phi(x)=\frac{1}{N} \left( I+\sum_{i=1}^{N^2-1} x_i P_{i+1} \right) \in D(\cV).$$

Notice that for any $\rho \in D(\cV)$, $\Phi(v(\rho))=\rho$ and
similarly, for any $x \in \KO$, $v(\Phi(x))=x$. Also for every $x
\in \R^{N^2-1}$ (not necessarily in $\KO$) we have that
$\Tr(\Phi(x))=1$, and for every $x,y \in \R^{N^2-1}$,
$\norm{\Phi(x) - \Phi(y)}_2 = \frac{1}{\sqrt{N}}\norm{x - y}_2$
where the first norm is over $L(\cV)$ and the second over
$\R^{N^2-1}$.

%\begin{claim}\label{cl:closeness}
%For every $x,y \in \R^{N^2-1}$, $\norm{\Phi(x) - \Phi(y)}_2 =
%\frac{1}{\sqrt{N}}\norm{x - y}_2$.
%\end{claim}
%
%\begin{proof}
%\begin{eqnarray*}
%\norm{\Phi(x)-\Phi(y)}_2^2 &=& \norm{ \frac{1}{N}\sum_{i=1}^{N^2-1} (x_i-y_i) P_{i+1}}_2^2\\
%&=& \frac{1}{N^2} \sum_{i=1}^{N^2-1} (x_i-y_i)^2
%\Tr(P_{i+1}^\dagger P_{i+1}) = \frac{1}{N} \norm{x-y}_2^2.
%\end{eqnarray*}
%\end{proof}

The convex set that we optimize over is $K=\KO \times \KO$. We
claim:

\begin{claim}
$K$ is convex and $B_{2N^2-2}(\ol{0},{1 \over 2\sqrt{N}})
\subseteq K \subseteq B_{2N^2-2}(\ol{0},2N)$.
\end{claim}

\begin{proof}
$\KO$ is convex since the set of density matrices is convex. Hence
$K$ is also convex. Next we show $B_{N^2-1}(\ol{0},{1 \over
2\sqrt{N}}) \subseteq \KO$ which implies $B_{2N^2-2}(\ol{0},{1
\over 2\sqrt{N}}) \subseteq K$. Indeed, let $x \in \R^{N^2-1}$ be
such that $\norm{x}_2 \le {1 \over 2\sqrt{N}}$ and let $$\rho
=\Phi(x)={1 \over N} \left(I+\sum_{i=2}^{N^2} x_i P_i\right).$$
Clearly $\rho$ is Hermitian and has trace $1$. We are left to
verify that $\rho$ is positive semi-definite. Fix a unit vector $u
\in \R^{N}$. Then,
\begin{eqnarray*}
u^\dagger \rho u &=& {1 \over N} \left( u^\dagger I u +\sum_{i=1}^{N^2-1} x_i u^\dagger P_{i+1} u \right) ~\ge~ {1 \over N} \left( 1 - \Big|\sum_{i=1}^{N^2-1} x_i u^\dagger P_{i+1} u\Big| \right) \\
& \ge & {1 \over N} \left( 1-\sum_{i=1}^{N^2-1} |x_i| \cdot \norm{P_{i+1}} \right)
%= 1-\sum_{i=1}^{N^2-1} |x_i|
~\ge~ {1 \over N} \left( 1-\sqrt{N} \norm{x}_2 \right) > 0.
\end{eqnarray*}

In order to show $K \subseteq B_{2N^2-2}(\ol{0},2N)$ it is enough
to show $\KO \subseteq B_{N^2-1}(\ol{0},N)$. Let $x \in \KO$. Then
$\rho =\Phi(x) \in D(\cV)$ and for any $1 \le i \le N^2-1$,
$$v_i(\rho)=|\Tr(\rho P_{i+1})| \le \Tr(|\rho P_{i+1}|) \le
\norm{P_{i+1}} \Tr(\rho) \le 1,$$ and so
$\norm{x}_2=\norm{v(\rho)}_2 \le N$.
\end{proof}

\begin{claim}
\label{cl:membership}
There exists an efficient membership oracle for $K$.
\end{claim}

\begin{proof}
Clearly it is enough to give an efficient membership oracle for
$\KO$. Given an input $x \in \R^{N^2 - 1}$ and an $\eps>0$ we
construct the Hermitian matrix $\rho=\Phi(x)$ and approximate its
eigenvalues with accuracy $\zeta=\frac{\eps}{10N^{3/2}}$ in the
$\ell_\infty$ norm. We then look at its smallest eigenvalue and we
return $1$ if it is positive and $0$ otherwise.

Given $x$, let $\sum_i \lambda_i \ketbra{v_i}{v_i}$ with
$\lambda_1 \ge \ldots \ge \lambda_N$ be the spectral decomposition
of $\rho=\Phi(x)$. The correctness of the membership oracle
follows from the following two claims:
\begin{itemize}
\item
If $x \in \KO_{-\eps}$ then $\lambda_N \ge {\eps \over 10 \sqrt{N}} >
\zeta$.
\item
If $x \not \in \KO_{+\eps}$ then $\lambda_N \le
-\frac{\eps}{2N^{3/2}} < -\zeta$.
\end{itemize}

For the first item, assume $x \in \KO_{-\eps}$ but $\lambda_N \le
{\eps \over 10 \sqrt{N}}$. Define
$\sigma=(1+\alpha)\rho-\alpha\ketbra{v_N}{v_N}$ for
$\alpha=\frac{2\lambda_N}{1-\lambda_N}$. Then $v(\sigma) \not \in
\KO$ because $\ket{v_N}$ is an eigenvector of $\sigma$ with
negative eigenvalue, but
$$\norm{x-v(\sigma)}_2=\sqrt{N}\norm{\rho-\sigma}_2 \le
\sqrt{N}\trnorm{\rho-\sigma} \le 2 \sqrt{N}\alpha
 \le 10 \sqrt{N}\lambda_N \le \eps,$$ and so $x \not \in \KO_{-\eps}$. A
contradiction.

For the second item, assume $x \not \in \KO_{+\eps}$ and $0 >
\lambda_N \ge -\frac{\eps}{2N^{3/2}}$. Define $\sigma={1 \over
1+\Delta}\sum_{i:\lambda_i>0} \lambda_{i}\ketbra{v_i}{v_i}$ for
$\Delta=-\sum_{i:\lambda_i<0} \lambda_i$. Clearly, $v(\sigma) \in
\KO$. Also, $$\norm{x-v(\sigma)}_2=\sqrt{N}\norm{\rho-\sigma}_2
\le \sqrt{N}\trnorm{\rho-\sigma} =2 \sqrt{N}\Delta \le 2 \sqrt{N}
N |\lambda_N| \le \eps.$$ Thus, $x \in
\KO_{+\eps}$. A contradiction.
\end{proof}

\subsection{The target function}

Let $\cV$ be a Hilbert space of dimension $N$. Let $T:L(\cV) \to
L(\cV)$ be a linear operator given in a Stinespring
representation, i.e., as a pair of operators $(B,C)$ such that
$$T(X) = \Tr_\cA (B X C^\dagger),$$ and let $\eps>0$. We assume
that $N$ is a power of $2$. From $B$ and $C$ we can compute $T_1$
and $T_2$ as in Equations~(\ref{eq:T1}) and~(\ref{eq:T2}). We
define a target function $g:K \to [-1,0]$ by
$$g(x,y)=-\sqrt{F}(T_1(\Phi(x)), T_2(\Phi(y))),$$

\begin{claim}
$g$ is convex over $K$.
\end{claim}

\begin{proof}
For every $0 \le \lambda_1,\ldots,\lambda_k \le 1$ such that
$\sum_{j=1}^k \lambda_j=1$,
\begin{eqnarray*}
g(\sum_{j=1}^k \lambda_j (\vj,\wj)) & = & g(\sum_{j=1}^k \lambda_j
\vj,\sum_{j=1}^k \lambda_j \wj) =
-\sqrt{F}(T_1(\Phi(\sum_{j=1}^k \lambda_j \vj)), T_2(\Phi(\sum_{j=1}^k \lambda_j \wj))) \\
& = & -\sqrt{F}(T_1(\sum_{j=1}^k \lambda_j \rho_j),
T_2(\sum_{j=1}^k \lambda_j \xi_j)),
\end{eqnarray*}
where $\rho_j=\Phi(\vj) \in D(\cV)$, $\xi_j=\Phi(\wj) \in D(\cV)$,
and we used the fact that $\Phi$ is linear for convex sums, i.e.,
$\Phi(\sum \lambda_j v_j)=\sum \lambda_j \Phi(v_j)$. Now, by the
joint concavity of $\sqrt{F}$,
\begin{eqnarray*}
g(\sum_{j=1}^k \lambda_j (\vj,\wj))
& = & -\sqrt{F}(\sum_{j=1}^k \lambda_j T_1(\rho_j), \sum_{j=1}^k \lambda_j T_2(\xi_j)) \\
& \le & -\sum_{j=1}^k \lambda_j \sqrt{F}(T_1(\rho_j), T_2(\xi_j))
= \sum_{j=1}^k \lambda_j g(\vj,\wj).
\end{eqnarray*}
\end{proof}% that g is convex over $K$

\begin{claim}
There exists an efficient evaluation oracle for $g$ over $K$.
\end{claim}

\begin{proof}
We are given as input $(x_1,x_2) \in K$ and $\eps>0$. We compute
$M_1=T_1(\Phi(x_1))$ and $M_2=T_2(\Phi(x_2))$ and this is done with no error. We would like to compute $g(x_1,x_2)=\trnorm{\sqrt{M_1} \sqrt{M_2}}$. We approximate $\sqrt{M_i}$ with $\zeta/2$ accuracy in the operator norm (it will turn out that $\zeta=\frac{\eps}{2 N \cdot (\norm{B}+\norm{C}+1)}$ suffices), and then we change each negative eigenvalue (if there are any) to zero. We get positive semi-definite $S_i$ such that $\norm{S_i-\sqrt{M_i}} \le \zeta$. We output an approximation of $\trnorm{S_1 S_2}$ with $\eps/2$ accuracy.

By Claims~\ref{cl:perturb-g} and~\ref{cl:perturb-norm} below:
\begin{eqnarray*}
\left|\trnorm{S_1 S_2} - \trnorm{\sqrt{M_1} \sqrt{M_2}}\right|
&\le& N \zeta \left(\norm{S_1}+\norm{\sqrt{M_2}}\right) \le
N \zeta (\norm{B}+\norm{C}+\zeta) \le \eps/2
\end{eqnarray*}
Thus, our output is $\eps$-close to $g(x_1,x_2)$ as required.
Also, observe that $\log (\zeta^{-1})$ is polynomial in the input
length, since $\log (\norm{B})$ and $\log (\norm{C})$ are
polynomial in the input length. Therefore, the evaluation oracle
is efficient.
\end{proof}

\begin{claim}\label{cl:perturb-g}
If $\rho_1,\rho_2,\sigma_1,\sigma_2 \in L(\cV)$ are positive
semi-definite and $\norm{\rho_i - \sigma_i} \le \zeta$ for
$i\in\set{1,2}$ then
$$\big| \trnorm{\rho_1 \rho_2} - \trnorm{\sigma_1 \sigma_2} \big|
\le N \zeta (\norm{\rho_1}+\norm{\sigma_2}).$$
\end{claim}

\begin{proof}
\begin{eqnarray*}
\big| \trnorm{\rho_1 \rho_2} - \trnorm{\sigma_1 \sigma_2}\big| &
\le &
\big| \trnorm{\rho_1 \rho_2} - \trnorm{\rho_1 \sigma_2}\big| +\big| \trnorm{\rho_1 \sigma_2} - \trnorm{\sigma_1 \sigma_2}\big| \\
&\le & \trnorm{\rho_1(\rho_2-\sigma_2)}+\trnorm{(\rho_1-\sigma_1)\sigma_2} \\
&\le& \norm{\rho_1}\trnorm{\rho_2-\sigma_2}+\norm{\sigma_2}\trnorm{\rho_1-\sigma_1} \\
& \le & N \zeta (\norm{\rho_1}+\norm{\sigma_2}).
\end{eqnarray*}
\end{proof}

\begin{claim}\label{cl:perturb-norm}
For any $\rho \in D(\cV)$: $\norm{\sqrt{T_1(\rho)}} \le \norm{B}$,
$\norm{\sqrt{T_2(\rho)}} \le \norm{C}$.
\end{claim}

\begin{proof}
$T_1$ is completely positive and so $T_1(\rho)$ is positive
semi-definite and
$\norm{\sqrt{T_1(\rho)}}=\sqrt{\norm{T_1(\rho)}}$. Express
$\rho=\sum_i \lambda_i \ketbra{v_i}{v_i}$ with $\set{\ket{v_i}}$
being an orthonormal basis, $\lambda_i>0$ and $\sum_i
\lambda_i=1$. Denote $\ket{w_i}=B \ket{v_i}$. Then,
$$\norm{T_1(\rho)} = \norm{\sum_i \lambda_i \Tr_\cV (B
\ketbra{v_i}{v_i} B^\dagger)} \le \sum_i \lambda_i \norm{\Tr_\cV
(\ketbra{w_i}{w_i})} \le \sum_i \lambda_i \norm{\tiny
\ket{w_i}\tiny}_2^2,$$ where we have used
$\norm{\Tr_{\cV}(\ketbra{w}{w})} \le \norm{\tiny
\ket{w_i}\tiny}_2^2$.
%
%To see that, express $w=\sum a_k \alpha_k \beta_k$ with %$\set{\alpha_k},\set{\beta_k}$ orthonormal bases and $a_k >0$. Then,
%$Tr_{\cV}(\ketbra{w}{w})=\sum_k |a_k|^2 \ketbra{\beta_k}{\beta_k}$ and
%$\norm{\Tr_{\cV}(\ketbra{w}{w})} \le \sum_k |a_k|^2 = \norm{w}^2$.
%
Thus, $\norm{T_1(\rho)} \le \norm{B}^2 \sum_i \lambda_i =
\norm{B}^2$. A similar argument applies for $T_2$.
\end{proof}

%\begin{eqnarray*}
%\trnorm{\sqrt{Q_1}-\sqrt{Q_1'}} & \le& \sqrt{N} %\norm{\sqrt{Q_1}-\sqrt{Q_1'}}_2 \\
%&\le& \sqrt{N \trnorm{|Q_1|-|Q_1'|}} \\
%&\le& \sqrt{N} \left( \sqrt{N} \norm{|Q_1|-|Q_1'|}_2 \right)^{1/2} \\
%&\le& N \sqrt{\norm{Q_1-Q_1'}_2 } \le N \trnorm{Q_1-Q_1'} \le N^2
%\norm{Q_1-Q_1'},
%\end{eqnarray*}
%where the second inequality follows from \cite[Page 298, Equation
%(X.23)]{Bhatia97} and the fourth inequality follows from
%\cite[Corollary VII.5.6]{Bhatia97}.

\subsection{The algorithm}

To compute the diamond norm of a given super-operator, the
algorithm essentially solves the convex program that finds the
minimum value of $g$ over the convex set. The last thing that we
need is to show that $g$ is indeed defined and can be evaluated
over points that are at most $\eps$-far from this set. However the
set $K$ is not good enough for this purpose since matrices that
lie outside this set (but still close to it) have negative
eigenvalues and it is not clear how one should define the fidelity
for such matrices. To overcome this problem we define a new convex
set $S$ that is just a shrinking of $K$. This ensures that
matrices that are $\eps$-close to the boundary are still positive.

We set $M=-N \sqrt{\norm{T_1} \norm{T_2}}$, where $\norm{T_i}$ is the spectral norm of $T_i$ when viewed as a linear operator in $\Hom(L(\cV),L(\cA))$. It can be verified
that $\min_{x \in K} g(x) \ge -M$.
%-------------
%For any $x_1,x_2 \in \KO$,
%
%\begin{eqnarray*}
%g(x_1,x_2) &=&
%-\trnorm{\sqrt{T_1(\Phi(x_1))} \sqrt{T_2(\Phi(x_2))}} \\
%& \ge &
%-\norm{\sqrt{T_1(\Phi(x_1))}} \trnorm{\sqrt{T_2(\Phi(x_2))}} \\
%&\ge &
%-\sqrt{N} \norm{\sqrt{T_1(\Phi(x_1))}} \norm{\sqrt{T_2(\Phi(x_2))}} \\
%& = &
%-\sqrt{N} \sqrt{\norm{T_1(\Phi(x_1))}}  \sqrt{\norm{T_2(\Phi(x_2))}} \\
%& \ge &
%-\sqrt{N} \sqrt{\norm{T_1}}  \sqrt{\norm{T_2}}  \sqrt{\norm{\Phi(x_1)}}  \sqrt{\norm{\Phi(x_2)}} \\
%& \ge &
%-\sqrt{N} \sqrt{\norm{T_1}}  \sqrt{\norm{T_2}}  \sqrt{\trnorm{\Phi(x_1)}}  \sqrt{\trnorm{\Phi(x_2)}} = -\sqrt{N} \sqrt{\norm{T_1}}  \sqrt{\norm{T_2}}
%\end{eqnarray*}
%-------------
%
Given $\eps>0$,
we define $\alpha={\eps \over 4M}$ and $\eps'={\alpha \over \sqrt{N}}$.
We define
$$\SO=(1-\alpha)\KO.$$

\begin{claim}
$\SO=\set{x \in K ~:~ \lambda_N(\Phi(x)) \ge {\alpha \over N}}$.
Furthermore, $\SO$ is convex, has an efficient membership oracle and $\SO_{+\eps'} \subseteq \KO$.
\end{claim}

\begin{proof}
\begin{eqnarray*}
z \in \SO & \Leftrightarrow & z=(1-\alpha)x \mbox{ for some $x \in \KO$}\\
 & \Leftrightarrow & \Phi(z)=(1-\alpha)\Phi(x)+\alpha \frac{I}{N} \mbox{ ~for some $\Phi(x) \in D(\cV)$} \\
 & \Leftrightarrow & \lambda_N(\Phi(z)) \ge {\alpha \over N}.
 \end{eqnarray*}

$\SO$ is convex and has an efficient membership oracle because
$\KO$ does. Also, $\SO_{+\eps'} \subseteq \KO$ because if $z \in
\SO$ and $\norm{x-z}_2\le \eps'$ then $$\lambda_N(\phi(z)) \ge
\lambda_N(\phi(x))-\norm{\Phi(x)-\Phi(z)} = {\alpha \over
N}-{\norm{x-z} \over \sqrt{N}} \ge {\alpha \over N}-{\eps' \over
\sqrt{N}} = 0.$$
\end{proof}

%\begin{claim}
%For any positive semi-definite $A,B$,  $|\lambda_{min}(A)-\lambda_{min}(B)| \le \norm{A-B}$.
%\end{claim}
%
%\begin{proof}
%W.l.o.g. assume $\lambda_{min}(A) \ge \lambda_{min}(B)$. Let $x$ be a unit vector such that $\lambda_{min}(B)=x^*Bx$. Then $\norm{A-B} \ge |x^*(A-B)x| =x^*Ax-x^*Bx \ge \lambda_{min}(A)-\lambda_{min}(B)$.
%\end{proof}
%

We are now ready to prove:

\begin{theorem}
Let $\cV$ be a Hilbert space of dimension $N$. Let $T:L(\cV) \to
L(\cV)$ be a linear operator given in a Stinespring
representation, i.e., as a pair of operators $(B,C)$ such that
$$T(X) = \Tr_\cA (B X C^\dagger),$$ and let $\eps>0$. Then there
exists a polynomial time algorithm (in the input length of $T$ and
$\log \eps^{-1}$) that outputs a value $c$ such that $
|~c-\dnorm{T}| \le \eps$.
\end{theorem}

\begin{remark}
The fact that the input operator $T$ is given in a Stinespring
representation is without loss of generality as there exists
efficient algorithms to move from such a representation to other
standard forms of representing a super-operator (see,
e.g.,~\cite[Lecture 5]{W04}).
\end{remark}

\begin{proof}
We approximate $\norm{T_i}$ from above in time polynomial in the representation of $T_i$,
and set $M, \alpha$, and $\eps'$ as above. We define $S=\SO \times \SO$
and $g:K \to \R$ as above. The target function $g$ has an
efficient membership oracle and is convex over $K$ and therefore
over $S_{+\eps'}$. By Theorem \ref{thm:convex-programming} we can
find a value $\widetilde{\opt}$ that approximates  $\min_{x \in
S_{-\eps'}} g(x)$ to within $\eps'$.

Now, let $o=(o_1,o_2) \in K$ be a point minimizing $g$ over $K$,
that is, $g(o) = \min_{x \in K} g(x)$. We claim that
$o'=(1-2\alpha)o$ lies in $S_{-\eps'}$. Indeed, fix any $y_i \in
B_{N^2-1}(o_i',\eps')$. Then, $$\lambda_N(\Phi(y_i)) \ge
\lambda_N(\Phi(o'_i))-{\eps' \over \sqrt{N}} \ge {2\alpha \over N}
-{\eps \over \sqrt{N}} = {\alpha \over N},$$ and therefore $y \in
S$. Thus,
$$g(o) \le \widetilde{opt} \le g(o')+\eps'.$$

However,
\begin{eqnarray*}
g(o') & = & g((1-2\alpha)o) =
-\sqrt{F}\left(T_1\left((1-2\alpha)\Phi(o_1)+2\alpha{I \over
N}\right),
T_2 \left((1-2\alpha)\Phi(o_2)+2\alpha {I \over N}\right)\right) \\
& \le &
(1-2\alpha)\left(-\sqrt{F}\big(T_1(\Phi(o_1)),T_2(\Phi(o_2))\big)\right)+
2\alpha \left(-\sqrt{F}\left(T_1\left({I \over N}\right),T_2\left({I \over N}\right)\right)\right) \\
& \le & (1-2\alpha)g(o_1,o_2)-2{\alpha \over N}
\sqrt{F}\left(\Tr_{\cV} BB^{\dagger},\Tr_{\cV} CC^{\dagger}\right)
\le (1-2\alpha)g(o).
\end{eqnarray*}

Altogether, $|\widetilde{opt}-g(o)| \le \eps'-2\alpha ~g(o) \le \eps'+2\alpha M \le \eps$.
\end{proof}

\bibliographystyle{alpha}
\bibliography{refs}
\end{document}